\documentclass[letterpaper,conference]{IEEEtran}
\usepackage{amsmath,amssymb,verbatim}
\usepackage[dvips]{graphicx}
\usepackage{setspace}
\usepackage{cite}
\usepackage{pb-diagram}

\usepackage{amssymb}
\usepackage{mathrsfs}

\usepackage{pb-diagram}
\usepackage{array}

\usepackage{enumerate}

\usepackage{color}

\newcommand{\x}{\times}

\usepackage{graphicx}

\usepackage{amssymb}
\usepackage{graphicx}

\usepackage{amssymb}
\usepackage{mathrsfs}

\usepackage{array}

\usepackage{enumerate}

\usepackage{color}

\newcommand{\F}{\mathbb F}

\newcommand{\OO}{\mathcal{O}}

\newcommand{\bv}{\mathbf{b}}
\newcommand{\gv}{\mathbf{g}}
\newcommand{\yv}{\mathbf{y}}

\newcommand{\Cc}{\mathcal{C}}

\newcommand{\R}{\mathbb{R}}
\newcommand{\C}{\mathbb{C}}

\newcommand{\Q}{\mathbb{Q}}
\newcommand{\Z}{\mathbb{Z}}

\newcommand{\qa}[3]{({#1},{#2})_{#3}}

\newcommand{\bm}[4]{\begin{bmatrix}{#1} &{#2}\\{#3} &{#4}\end{bmatrix}}

\newtheorem{lemma}{Lemma}

\newtheorem{remark}[lemma]{Remark}

\newcommand{\alphatau}[1] {\begin{bmatrix} X &0 & 0 \\ 0 &{\tau(X)} & 0 \\ 0 &0 & {\tau^2(X)} \end{bmatrix}}

\newtheorem{definition}{Definition}[section]
\newtheorem{thm}{Theorem}[section]
\newtheorem{proposition}[thm]{Proposition}

\newcommand{\tr}{{\rm Tr}}

\newcommand{\G}{{\Gamma}}
\renewcommand{\G}{{\mathcal J}}

\begin{document}
 
\title{Algebraic Fast-Decodable Relay Codes for Distributed Communications}

\author{\IEEEauthorblockN{Camilla Hollanti, \emph{Member, IEEE}} 
\IEEEauthorblockA{Department of Math. and Syst. An.\\  P.O. Box 11100\\FI-00076 Aalto University\\ Finland\\ camilla.hollanti@aalto.fi\\}
\and
\IEEEauthorblockN{ Nadya Markin}  
\IEEEauthorblockA{ School of Phys. and Math. Sciences \\
Nanyang Technological University\\
21 Nanyang Link \\
Singapore 637371 \\
nadyaomarkin@gmail.com \\
}
}

\maketitle

\begin{abstract}
In this paper, fast-decodable lattice code constructions are designed for the nonorthogonal amplify-and-forward (NAF) multiple-input multiple-output (MIMO) channel. The constructions are based on different types of algebraic structures, \emph{e.g.} quaternion division algebras. When satisfying certain properties, these algebras provide us with codes whose structure naturally reduces the decoding complexity.  The complexity can be further reduced by shortening the block length, \emph{i.e.}, by considering rectangular codes called less than minimum delay (LMD) codes. 
\end{abstract}


\section{Introduction}

The quality of wireless long distance communications can be significantly improved by using cooperative diversity techniques. Cooperating relays can be positioned between the source station and the destination to aid the transmission by either amplifying and forwarding (AF) or decoding and forwarding (DF) the signal. Spatially separated terminals will allow an increment in the diversity in a distributed manner. Depending on the application, a one-hop or multi-hop transmission is called for. Here, we consider multi-hop distributed space-time codes employing a half-duplex NAF protocol \cite{nabar_naf}, \cite{azarian_naf}.  It is known \cite{azarian_naf} that the NAF protocol outperforms all other AF protocols since, as opposed to orthogonal protocols, it can keep transmitting also during the transmission of the relays. In addition, the AF protocols are less complex than the DF protocols. This type of low cost relay systems are called for in \emph{e.g.} digital video broadcasting (DVB) \cite{DVB}.

In \cite{belfi_naf} and \cite{asykonstru}, Yang \emph{et al.} and Hollanti \emph{et al.} proposed block-diagonal space-time code constructions for the asymmetric MIMO channel with or without relays. The constructions arise from cyclic division algebras constructed over a higher degree center. A nonvanishing determinant  (NVD) is then achieved by forming a block-diagonal matrix consisting of the left regular representation of the algebra and its Galois conjugates from the center to the base field. It was also shown \cite{belfi_naf} that a block-diagonal structure together with the NVD property is enough to achieve the diversity-multiplexing gain tradeoff (DMT) also in the asymmetric case, where the number of transmit antennas is strictly bigger than the number of receive antennas, and hence the corresponding lattice is not full. Motivated by this and the urge for complexity reduction of MIMO codes in general, we impose further properties that the algebras and the constructions should satisfy in order to reduce the complexity. Our study reveals a trade-off between the coding gain and decoding complexity. Related work has been carried out by, among others,  Rajan \emph{et al.} (see e.g. \cite{RR}). They considered fast-decodable distributed space-time codes arising from Clifford algebras. Our work differs from theirs in that our codes achieve the NVD property and hence the asymmetric DMT.  The codes proposed in this paper moreover have a nice algebraic structure which makes analyzing the codes easier.
  
List of contributions:
\begin{itemize}
\item Explicit fast-decodable space-time relay codes are proposed. 
\item All the codes have full diversity, some even NVD. To the best of the authors' knowledge, these are the first fast-decodable distributed space-time codes with NVD. 
\item  The constructions arise from quaternion or other type of algebras making it easy to determine the coding gain, complexity and other properties of the codes. 
\item Although our explicit examples are for the case when the source and the relays each have only one antenna, the  constructions are easily generalizable to other numbers of antennas and relays as well. 
\item We demonstrate a performance-complexity tradeoff resulting from the used method.
\item Finally, we analyze the worst-case decoding complexity of the proposed codes. 
\end{itemize}

 Let us finish this introductory section by giving a couple essential definitions. 
 
 \begin{definition}
 If the code $\Cc$ consisting of matrices $X$ satisfies
 $$\min_{0\neq X\in \Cc}\det(X^\dag X)>\kappa>0,$$
 we say that $\Cc$ has the \emph{nonvanishing determinant property} (NVD).
 
 In case of square matrices, we simply refer to $\det(X)$ when talking about NVD.
  \end{definition}
 
 There are multiple definitions of rate, but we will consistently use the following. 
 
 \begin{definition}
 Let $B_1,\ldots,B_k\in M_{n_t\times T}(\C)$ be the generator matrices (over $\R$) of a rank $k$ code $\Cc$, so
 $$
 \Cc=\sum_{i=1}^k B_ig_i,
 $$
 where $g_i\in\Z$, \emph{e.g.} PAM symbols.
 The \emph{rate} $R$ of the code is then 
 $$
 R=k/T
 $$
(real) dimensions per channel use (dpcu).
  \end{definition}
 
 Note that the commonly used rate in complex dimensions per channel use is $R/2$ when using the above notation.
 
  \section{System model for the NAF relay channel}
  
  For ease of notation, we only define the single-relay model, the generalization to multi-hop is straightforward.  Following \cite{belfi_naf}, let us denote by $X_i$ the signals transmitted from the source, and by $Y_r$ the signal received by the relay which the relay then amplifies and forwards as $BY_r$. The number of relays and the number of antennas at the source, relays and destination are denoted by $N,n_s,n_r,n_d$, respectively. We assume $n_r$ is the same for all relays $r=1,\ldots, N$. To be realistic, we assume $n_s\geq n_r$. The destination is observing $Y_1$ and $Y_2$ in consecutive time instances, and we have
  
  \begin{eqnarray*}
  Y_1&=&\sqrt{\pi_1\,SNR}\,FX_1+V_1\\
    Y_r&=&\sqrt{\pi_1\rho\, SNR}\,HX_1+W_1\\
      Y_2&=&\sqrt{\pi_3\,SNR}\,G(BY_r) +\sqrt{\pi_2\,SNR}\,FX_2+V_2,\\
      \end{eqnarray*}
where $V_i,\,W$ are the additive white gaussian noise matrices and $F,\, H,\, G$ are the Rayleigh distributed channel matrices. The power allocation $\pi_i$ factors are chosen so that $SNR$ denotes the received SNR per receive antenna at the destination. We assume perfect channel state information (CSI) at the receivers, while the transmitters have none.  For more details, we refer to \cite{belfi_naf}.

%
%

\section{On fast decodability}

Maximum-likelihood decoding amounts to searching the code $\Cc$ for the codeword
\begin{equation}
\label{frob-min}
Z = argmin\{||Y-HX||_F^2\}_{X \in \Cc},
\end{equation} closest to the received matrix $Y$ with respect to the squared Frobenius norm. 

Consider a code $\Cc$ of $\Q$-rank $k$, \emph{i.e.}, each codeword $X$ is a linear combination $\sum_{i=1}^kB_ig_i$ of generating matrices $B_1, \ldots, B_{k}$, weighted by coefficients $g_1,\dots,g_{k}$, which are PAM information symbols. The matrices $B_1, \ldots, B_{k}$ therefore define our code.
Each $n_r\times T$ matrix $HB_i$ corresponds, via vectorization, to a vector $\bv_i \in \R^{2Tn_r}$ obtained by stacking the columns followed by separating the real and imaginary parts of $HB_i$. We define the (generating) matrix 
\[
B=(\bv_1, \bv_2,\dots, \bv_{k}) \in M_{{2 T n_r}\times k}(\mathbb R),
\] so every received codeword can be represented as a real vector $B\gv$, with $\gv=(g_1,\dots,g_{k})^T$ having coefficients in the real alphabet $S$ in use. 

Now finding 
$argmin\{||Y-HX||_F^2\}_{X \in \Cc}$ 
becomes equivalent to finding $argmin\{|| \yv -B\gv ||_E^2\}_{\gv \in |S|^{k}}$ with respect to Euclidean norm, where $\yv$ is the vectorization of the received matrix $Y$.
The latter search is performed using a real sphere decoder \cite{VB}, with the complexity of exhaustive search amounting to $|S|^{k}$, as the coefficients of $\gv$ run over all the values of $S$. The complexity of decoding can, however, be reduced if the code has additional structure \cite{JR}. Performing a QR decomposition of $B$, $B=QR$, with $Q^\dagger Q=I$, reduces finding $argmin\{|| \yv -B\gv ||_E^2\}_{\gv}$ to minimizing 
\begin{equation}\label{eq:dR}
|| \yv-QR\gv||_E^2=||Q^\dagger\yv-R\gv||_E^2,
\end{equation}
where $R$ is an upper right triangular matrix. 
The number and structure of zeros of the matrix $R$ may improve the {\emph{decoding complexity}} (formally defined \cite{BHV} to be the minimum number of vectors $\gv$ over which the difference in (\ref{eq:dR}) must be computed). When the structure of the code allows for the degree (\emph{i.e.}, the exponent of $|S|$) of decoding complexity to be less than the rank of the code, we say that the code is {\emph{fast-decodable}}.

More precisely, we have the following definitions from \cite{JR}.  

\begin{definition} \label{FDdef}A space-time code is said to be {\emph {fast-decodable}} if its $R$ matrix has the following form:
$$R = {\bm {\Delta} {B_1} 0 {R_2}},$$ \label{eq:Rcondition}
where $\Delta$ is a diagonal matrix and $R_2$ is upper-triangular. 
\end{definition}

The authors of \cite{JR} give criteria when the zero structure of $R$ coincides with that of $M$, where $M$ is a matrix capturing information about orthogonality relations of the basis elements of $B_i$:

\begin{equation}M_{k,l}= ||B_k^\dagger B_l+B_l^\dagger B_k||_F. \label{orthRelations}\end{equation}

In particular, \cite[Lemma 2]{JR} shows that if $M$ has the structure 
$M = \bm {\Delta} {B_1} {B_2} {B_3}$, where $\Delta$ is diagonal, then $R = \bm {\Delta} {B_1} {0} {R_1}$. We could thus rephrase Defnition \ref{FDdef} in terms of $M$.

Next we recall the class of codes which allows groups of symbols to be decoded independently of one another.  

\begin{definition}
A space-time code of dimension $K$ is called {\emph{g-group decodable}} if there exists a partition of $\{1, \ldots, K\}$ into $g$ nonempty subsets $\G_1, \ldots, \G_g$, so that the matrix $M_{l,k}=0$ when $l, k$ are in disjoint subsets $\G_i , \G_j$. \end{definition}

In this case, as shown in \cite{JR}, the matrix $R$ has the form 
$R = \begin{bmatrix}{R_1} & {0} & {0}\\
{0} & {\cdots} & {0}\\
{0} & {0} & {R_g}\end{bmatrix}$ \label{eq:RconditionGD}

where each $R_i$ is a square upper triangular matrix. Hence, the symbols $x_k$ and $x_l$ can be decoded independently when their corresponding basis matrices $B_k$ and $B_l$ belong to disjoint subsets of the partition.

\begin{remark}Note that a simple computation shows that the zero structure of $M$ is stable under premultiplication of $B_i $ by a channel matrix $H$. In general, the same does not hold for $R$.
\end{remark}

By the above discussion, in order to demonstrate fast-decodability (resp. g-group decodability), it suffices to find an ordering on the basis elements $B_i$, which results in the desired zero structure of $M$. We proceed to do that for the proposed relay codes. 


\section{Minimum delay codes for $N=2$ and $N=3$ relays}
We demonstrate that the codes we obtain are conditionally $4$-group decodable. Recall from \cite{NR} that a code is called \emph{conditionally $g$-group decodable} if there exists a partition of the indices $\{1,\ldots,K\}$ of basis elements into $g+1$ disjoint subsets $\G_1$, \ldots, $\G_g$, $\G^C$ such that

$$\|B_l^\dag B_m+B_m^\dag B_l \|_F =0 \quad \forall l \in \G_i, \forall m \in \G_j, i \neq j.$$

In this case, the sphere decoding complexity order reduces to $|S|^{{|\G^C|}+\max_{1 \leq i \leq g} {|\G_i|}}$.

\subsection{Virtual $6\x2 $ MIMO channel with  $N=3, n_s=n_r=1, \ n_d=2$}
We proceed to show the rate four (4 PAM symbols per channel use) relay construction consisting of $6\x6$ matrices. 
\begin{proposition}
Define the code $$\Cc = \left\{\alpha_\tau(X)\right\}=\left\{\alphatau X\right\} $$ where $X$ is a matrix of the form  $$X =\left(\begin{matrix} c &  -\sqrt{11}\sigma(d)\\ \sqrt{11}d & {\sigma(c)}\end{matrix}\right)$$ with 
$c, d \in \Z(i, \zeta_7)$, $\sigma: i\mapsto -i$.

Then the code $\Cc$ is of rank 24 and (real) decoding complexity $|S|^{15}$, and has the NVD property. 
\end{proposition}

\begin{IEEEproof}
Define $K = \Q(\sqrt{-7})$, $K' = \Q(\zeta_7)$.
Let $\zeta$ denote $\zeta_7$, then $K' = K(\zeta+\zeta^{-1})$. 

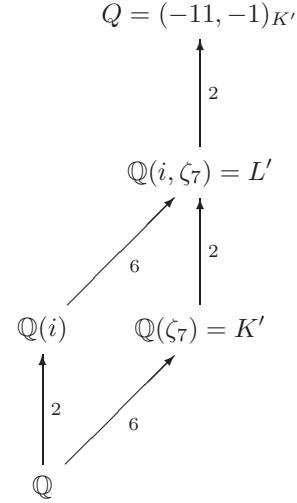
\[
\begin{diagram}
\node{} \node{Q = \qa {-11}{-1}{K'}}\\
\node{} \node{\Q(i, \zeta_7) = L'}\arrow{n,r}{2}\\
\node{\Q(i)}\arrow{ne,r}{6} \node{\Q(\zeta_7)=K'}\arrow{n,r}{2}\\
\node{\Q}\arrow{n,r}{2}\arrow{ne,r}{6}\\
\end{diagram}
\]

First we note that the algebra $\qa {-11} {-1} {\Q(\zeta_7)}$ is division. This follows from techniques of \cite{UM}: we apply \cite[Theorem 7.1]{UM} while noting that $\F_{11^3}$ contains no element of order $4$, \emph{i.e.},$-1$ is not a square in $\F_{11^3}$, which is the residue field of the prime $11$ in $\Q(\zeta_7)$. 

Now note that $\qa {-11}{-1}{K'} \cong \qa{-1}{-11}{K'}$, so let $Q = \qa{-1}{-11}{K'}$. 
After conjugation that does not affect the determinant but does aid energy balance and decoding complexity, it has the following $K'$-basis:
$$\left\{q_1 = \bm 1 0 0 1, q_2 = \bm {i} 0  0 {-i} ,\right. $$ 
$$\left. q_3 = \bm 0 {\sqrt{11} i} {\sqrt{11} i} 0,   q_4 = \bm 0 {-\sqrt{11}} {\sqrt{11}} 0 \right\}.$$
This means that $Q$ is generated over $K=\Q(\sqrt{-7})$ by the following 12 matrices
$$\Gamma_{i,1}=q_i, \Gamma_{i,2}=q_i(\zeta+\zeta^{-1}), \Gamma_{i,3}=q_i(\zeta^2+\zeta^{-2})$$ for $i = 1, \ldots, 4.$ 
Extend this to a $\Q$-basis by letting $\Gamma_{i,j}= \sqrt{-7}\Gamma_{i-4, j}$ for $i = 5, \ldots, 8$. Then a $\Z$-basis of $\Cc$ can be given by 

\begin{equation} \{ \alpha_\tau(\Gamma_{i,j}) \}_{i\leq 8,j\leq 3} \label{Zbasis6x6}\end{equation}
and is of size 24. Indeed, the rank of $\Cc$ is 24, since each codeword $X$ is an element of quaternion algebra $\qa {-1}{-11}{K'}$, and hence encodes 4 symbols from $K' = \Q(\zeta_7)$, or equivalently 24 real symbols.

Now let $\tau: \zeta_7 \mapsto \zeta_7^2$ be a generator of $Gal(K'/K)$.  When the coefficients of codewords are algebraic integers, the code is NVD. This follows from the fact that the determinant of each codeword is fixed by both $\tau$ and $\sigma$, hence it is an element of $\Q(\sqrt{-7})=K$. Moreover it is nonzero whenever $X$ is nonzero, since $Q$ was shown to be division.

We show that  $\Cc$ is conditionally $4$-group decodable with complexity $|S|^{15}$; conditioned on decoding symbols corresponding to $\{\Gamma_{5,1}, \ldots, \Gamma_{8, 3}\}$, the complexity of decoding symbols corresponding to $\{\Gamma_{1,1}, \ldots, \Gamma_{4,3}\}$ is at most $|S|^3$, where $S$ is the underlying alphabet. For that, note that when $A = \Gamma_{i,j}, B  = \Gamma_{i', j'} $, for all $j,j'$ and for $i \not = i'$, we have
$$AB^\dagger + BA^\dagger = {\bf 0}.$$

Same follows for $\alpha_\tau(A), \alpha_\tau(B)$, \emph{i.e.}, we have: 
$$\alpha_\tau(A)\alpha_\tau(B)^\dagger + \alpha_\tau(B)\alpha_\tau(A)^\dagger = {\bf 0}.$$

Let $\Gamma = [\alpha_\tau(\Gamma_{1,1}), \ldots , \alpha_\tau(\Gamma_{8,3})]$ be the list of $24$ generators of $\Cc$ from (\ref{Zbasis6x6}) in lexicographical order. Then the matrix  $M = M_{i,j}$ from Equation (\ref{orthRelations}) capturing orthogonality relations on $\Gamma$ has the following structure:

\begin{equation}M =\begin{bmatrix}*&0&0&0&*&*&*&* \\ 0&*&0&0&*&*&*&* \\ 0&0&*&0&*&*&*&* \\ 0&0&0&*&*&*&*&* \\ *&*&*&*&*&*&*&* \\ *&*&*&*&*&*&*&* \\ *&*&*&*&*&*&*&* \\ *&*&*&*&*&*&*&* \\ \end{bmatrix} \label{eq:M}\end{equation}
where each coefficient of the matrix above is a $3\x3$ matrix, which is $\mathbf{0}$ when the coefficient is $0$. 

\end{IEEEproof}

\subsection{Virtual $4\times 2$ MIMO channel with $N=2, n_s=n_r=1$, $n_d=2$} 
We use a similar idea as in the $6\x 6$ case to construct fast-decodable rate four relay  codes consisting of  $4\x 4$ matrices. 
\begin{proposition}
Let the code $\Cc$ consist of matrices $$\left\{\bm X 0 0 {\tau(X)} \bigg|\, X = {\footnotesize{\bm c  {-\sqrt{2} \sigma(d)} {\sqrt{2}d} {\sigma(c)}}}  \right\}$$ with 
$c, d \in \OO_K$, $\sigma: \sqrt{a}\mapsto -\sqrt{a}$. Again we have conjugated the original matrix in order to aid decodability and energy balance. The minimum determinant is invariant under such conjugation.

Each matrix $X$ represents an element from the quaternion algebra $\qa a \gamma K$ over a biquadratic field $K = \Q(i, \sqrt{p})$. 
Then for values $a=5, \gamma = -2, p =31 $ the resulting code is a fully diverse NVD code of rank 16. It  is conditionally $4$-group decodable with (real) decoding complexity $|S|^{10}$. 
\end{proposition}
\begin{IEEEproof}
We use similar techniques to the previous proof to establish that $\qa a \gamma K$ is division. First we establish that the prime ideal $5\Z$ splits completely in $K$. Using the fact that its residue field is isomorphic to $\F_5$, we conclude that $\gamma$ is not a square modulo $5$ in the integers of $K$. Hence using \cite[Theorem 7.1]{UM}, we conclude that $Q = \qa a \gamma K$ is division. 
The generators of $Q$ over $\Q(i)$ are 

$$\Gamma_{i,1}=q_i, \Gamma_{i,2}=q_i\sqrt{p} \mid i = 1, \ldots, 4.$$

We use the image of these generators under $\alpha_\tau$ to generate the code $\Cc$ over $\Z$. To check fast decodability, we verify the relations
$$\alpha_\tau(X)^\dagger \alpha_\tau(Y) + 
\alpha_\tau(Y)^\dagger \alpha_\tau(X) = {\bf{0}}, $$
$X =\Gamma_{i,j}, Y = \Gamma_{i',j'}$, where $1\leq i\not = i' \leq 4$ and $1\leq j, j' \leq 3$. 

The ordering $\alpha_\tau(\Gamma_{1,1}), \ldots, \alpha_\tau(i\Gamma_{1,1}) , \ldots, \alpha_\tau(i\Gamma_{4,2})$ gives the matrix $M$ with the same zero structure as in Equation (\ref{eq:M}) only now each coefficient is a $2\times 2$ matrix. Hence the code is conditionally $4$-group decodable with complexity $|S|^{10}$. More precisely, conditioned on decoding symbols corresponding to  $\{\alpha_\tau(i\Gamma_{1,1}), \ldots, \alpha_\tau(i\Gamma_{4,2}) \}$, the complexity of  decoding the rest of the symbols is at most $|S|^2$. 


\end{IEEEproof}

\begin{remark} Here we have concentrated on a real sphere decoding process. 
Note, however, that $\Gamma_{1,1}, \Gamma_{1,2}, \ldots , \Gamma_{4,1}, \Gamma_{4,2}$ gives a $\Z(i)$-basis of $\Cc$. We can verify the relations 
$$\alpha_\tau(X)^\dagger \alpha_\tau(Y) + 
\alpha_\tau(Y)^\dagger \alpha_\tau(X) = 0, $$
$X = \Gamma_{i,j}, Y = \Gamma_{i',j'}$, where $1\leq i\not = i'\leq 4$. 
Hence the obtained code is $4$-group decodable, of complexity degree $2$ when using a complex decoder. 
\end{remark}

\section{Less than minimum delay codes for $N=2$ relays}

In order to further reduce the complexity, we will shorten the block length. Such \emph{less than minimum delay} codes have been considered in \cite{lessthanmin}. To this end, let us start by constructing a code with dimension rate $R=2$, i.e., the lattice is of rank 4 and the code matrix transmits two real dimensions (e.g. two PAM symbols) per channel use. 

\subsection{Virtual $4\x1 $ MIMO channel with  $N=2, n_s=n_r=1, \ n_d=1$}
Let us consider the 8th cyclotomic extension $\Q(\zeta_8)/\Q$, and denote its Galois group by $\{1,\tau:\sqrt{2}\mapsto -\sqrt{2},^*,\tau ^*\}$, where $^*$ denotes complex conjugation. Our code will simply consist of matrices 
$$
X(a_1,a_2)=\left(\begin{array}{cc}
x & 0\\
x^* &0\\
0&\tau(x)\\
0&\tau(x)^*\\
\end{array}\right),
$$
where $x=a_1+a_2\zeta_8$ and $a_i\in\Z[i]$. Due to the fact that the code matrix only contains two QAM symbols, the complexity will automatically be at most $|S|^4$, where $S$ is the underlying real alphabet. However, if we do a smart ordering of the basis elements as 
$
\{B_1=X(1,0),X(i,0),X(0,1),X(0,i)\},
$
the $4\times 4$ matrix $R$ will have the form described in Definition \ref{FDdef} with a $2\times 2$ matrix $\Delta$. This is due to the fact that whenever we have a totally real basis element $B_i$ and a totally imaginary one $B_j$, this will result in $\Re\tr[(HB_i)^\dag HB_j]=0$. Thus, we have reduced the worst case complexity to $|S|^3$, that is, by 62.5\% compared to the complexity $|S|^8$ of a general  square code with the same rate.

\begin{proposition}
The matrices $X^\dag X,\, X\neq 0,$ have the NVD property.
\end{proposition}
\begin{proof}
The proof is straightforward. Namely, $\det(X^\dag X)=4x\tau(x)x^*\tau(x)^*=4N_{\Q(\zeta_8)\Q}(x)\in 4\Z$, and hence $\det(X^\dag X)\geq 4$.
\end{proof}

Next, we extend the above construction to the case $n_d=2$ ideally calling for a rate four code. 

\subsection{Virtual $4\x2 $ MIMO channel with  $N=2, n_s=n_r=1, \ n_d=2$}
Let us next construct a rank 8 lattice in order to have higher multiplexing of 4 real dimensions per channel use. We start by adjoining $\sqrt{5}$ to the above extension, \emph{i.e.}, we consider $\Q(\zeta_8,\sqrt{5})/\Q$ and denote the corresponding maps fixing $\Q(i)$ by 
   $   \{1,\tau:\sqrt{2}\mapsto -\sqrt{2},r:\sqrt{5}\mapsto -\sqrt{5},\tau r\}.$
   
 The code matrix now looks like
$$
 X(a_1,a_2,a_3,a_4)=\left(\begin{array}{cc}
\nu x & 0\\
 r(\nu x) &0\\
0&\tau(\nu x)\\
0&\tau r(\nu x)\\
\end{array}\right),
$$
where $x=a_1+a_2\zeta_8+a_3\theta+a_4\zeta_8\theta$,  $\theta=\frac{1+\sqrt{5}}2$,  $a_i\in\Z[i]$, and $\nu=1+i-i\theta$ generates a principal ideal that will make the code lattice orthogonal. This field extension is the same as the one used for the extended golden algebra in \cite{belfi_naf}. The complexity of the code is at most $|S|^8$. Similarly to the $n_d=1$ case, this can be further reduced to $|S|^7$ by ordering the basis as $\{B_1=X(1,0,0,0),B_2=X(i,0,0,0),\ldots\}$. We have reduced the complexity by 56.25\% compared to a general  square code with the same rate and complexity $|S|^{16}$. 

\begin{proposition}
The matrices $X^\dag X,\, X\neq 0,$ have full rank, \emph{i.e.}, a code consisting of the matrices $X$ has full diversity.
\end{proposition}
\begin{proof}
Again, the proof is very simple. We have $\det(X^\dag X)=(|\nu x|^2+|r(\nu x)|^2)(|\tau(\nu x)|^2+|\tau r(\nu x)|^2)>0$.
\end{proof}

\begin{remark}
We want to point out that, unfortunately, there is no free lunch. Namely with the above construction method increasing the code rate causes a degradation in the coding gain. At the same time, reducing the delay from four to two channel uses will naturally lower the maximum rank (and diversity) the code matrix can achieve when compared to a square code matrix. Hence, we observe a performance-complexity tradeoff implying that while we can indeed reduce the decoding complexity by reducing the code length, we are  likely to face slightly worse performance caused by the reduction in diversity.  Increasing the rate (performance) will also here require dropping the NVD requirement, again indicating a tradeoff type behavior. 
\end{remark}
\section{Conclusions and future work}\label{conclusions}

We proposed explicit relay codes with a fast-decodable structure and NVD for different number of antennas and relays. It was shown that the method used implies a performance-complexity tradeoff. In other words, fast-decodable codes with NVD were proposed, while at the same time it was noted that relaxing on the NVD property allows for further complexity reductions. One efficient way to reduce the complexity is to employ less than minimum delay codes that by construction already halve the complexity compared to general minimum delay codes. It remains to be investigated how the codes perform compared to other distributed codes that either have higher complexity or lack NVD. Related results \cite{caminadya} will be posted to arXiv in near future.



\end{document}